\RequirePackage{snapshot} 
\documentclass[10pt,conference]{IEEEtran}%
\usepackage{amssymb}
\usepackage{amsfonts}
\usepackage{amsmath}%
\usepackage[usenames,dvipsnames]{color}

\ifCLASSINFOpdf 
\usepackage[pdftex]{graphicx}  %remove draft option to print figures 
\usepackage[pdftex,colorlinks,%draft, %naturalnames,% 
           bookmarks=true,%
           ]{hyperref}  
\usepackage[numbers,sort&compress]{natbib} 
\usepackage{hypernat} 
\else

 \usepackage[dvips]{graphicx} 
\usepackage[dvips, colorlinks,draft, %naturalnames,% 
           bookmarks=true,%
           ]{hyperref}
 \usepackage[numbers,sort&compress]{natbib} 
\fi 
\hypersetup{
   bookmarksnumbered,
   pdfstartview={FitH},
   citecolor={blue},
   linkcolor={red},
   urlcolor=[rgb]{0,0.55,0},%{green},
   pdfpagemode={UseOutlines},
   pdftitle={Performance\ laws\ of\ large\ heterogeneous\  cellular\ networks},
   pdfauthor={B.\ Blaszczyszyn,\ M.\ Jovanovic,\ M.\ K.\ Karray},
   pdfsubject={},
   pdfkeywords={}
}

\graphicspath{{./}}

\newtheorem{thm}{Theorem}
\newtheorem{cor}[thm]{Corollary}
\newtheorem{lem}[thm]{Lemma}
\newenvironment{lemma}{\bf\begin{lem}\rm\em}{\end{lem}} % corollary 
\newtheorem{prop}[thm]{Proposition}
\newenvironment{proposition}{\bf\begin{prop}\rm\em}{\end{prop}} % corollary 
\newtheorem{rem}[thm]{Remark}
\newenvironment{remark}{\bf\begin{rem}\rm}{\end{rem}} % Numbered remark 

\newcommand{\SINR}{\mathrm{SINR}}
\newcommand{\Pro}{\mathbb{P}}

\newcommand{\Ex}{\mathbb{E}}

\newcommand{\calS}{\mathcal{S}}
\newcommand{\barN}{\bar N}
\newcommand{\ind}{\mathbf{1}\!}

\newcommand\exclude[1]{}

\begin{document}

\title{\huge Performance laws of large heterogeneous cellular networks}
\author{\IEEEauthorblockN{Bart{\l }omiej~B{\l}aszczyszyn\IEEEauthorrefmark{1},
Miodrag Jovanovic\IEEEauthorrefmark{2}\IEEEauthorrefmark{1} and
Mohamed Kadhem Karray\IEEEauthorrefmark{2}}}
\maketitle

\let\thefootnote\relax\footnote{\IEEEauthorrefmark{1}INRIA-ENS,
23 Avenue d'Italie, 75214  Paris, France
Email: Bartek.Blaszczyszyn@ens.fr\\
\indent\IEEEauthorrefmark{2}Orange Labs;
38/40 rue G\'{e}n\'{e}ral Leclerc, 92794  
Issy-les-Moulineaux, France
Email: \{miodrag.jovanovic,\,mohamed.karray\}@orange.com
%\indent This paper reports the results of a research undertaken under the
%contract  CRE N${}^\text{o}\,$D09262 between Inria and France T\'el\'ecom.
}

\newcommand{\thefootnote}{\arabic{footnote}}
\addtocounter{footnote}{-1}

\begin{abstract}
We propose a model for heterogeneous cellular networks assuming a space-time Poisson process of call arrivals, independently marked by data volumes, and served by different types of base stations (having different transmission powers) represented by the superposition of independent Poisson processes on the plane.  Each station applies a processor sharing policy to serve users arriving in its vicinity, modeled by the Voronoi cell perturbed by some random signal propagation effects (shadowing). Users' peak service rates depend on their signal-to-interference-and-noise ratios (SINR) with respect to the serving station. The mutual-dependence of the cells (due to the extra-cell interference) is captured via some system of cell-load equations impacting the spatial distribution of the SINR. We use this model to study in a semi-analytic way (involving only static simulations, with the temporal evolution handled by the queuing theoretic results) network performance metrics (cell loads, mean number of users) and the quality of service perceived by the users (mean throughput) served by different types of base stations.  Our goal is to identify macroscopic laws regarding these performance metrics, involving averaging both over time and the network geometry.  The reveled laws are validated against real field measurement in an operational network.
\end{abstract}


\section{Introduction}
Wireless cellular networks are constantly evolving to cope with the accelerating increase
of the traffic demand. The technology progressed  from 3G enhancement with HSDPA
(High-Speed Downlink Packet Access) to 4G with LTE (Long Term Evolution). The
networks become also more dense and more heterogeneous;  i.e. new base
stations (BS) of different types are added. In particular,  
operators introduce \emph{micro} BS, which transmit with smaller
powers  than the original ones (called \emph{macro} BS) in order to cope with local
increase of the traffic demand  (hotspots). The reasons for using smaller
transmitting powers is to avoid a harmful increase of interference
and reduce energy consumption as well as human exposure to the electromagnetic radiation. 
The deployment of micro BS is expected to increase significantly in
the nearest future.

Usage of different tiers of  BS (as the micro and  macro stations)  with variable transmission powers as well as
antenna gains, height etc, makes cellular networks heterogeneous.
Besides, even the macro tiers in commercial cellular networks are never perfectly regular: 
the locations  of  BS is usually far from being perfectly hexagonal,
because of various deployment constraints. Irregularity of the spatial patterns of BS
is usually more pronounced  in dense urban environments. Physical
irregularity of the urban environment (shadowing) induces  additional
variability of radio conditions.
Irregularity and heterogeneity  of cellular networks implies a spatial
disparity of base station performance metrics and quality of service (QoS) parameters observed by users 
in different cells of the network.
This  represents a challenge 
for the network operators, in particular in the context of the network dimensioning.
How to describe and analyze the performance of a large, irregular, heterogeneous
network? Which tier in the given network disposes larger
capacity margins? Is it the macro tier since its BS transmit with
larger powers  or the micro tier, whose BS  serve  smaller zones?
The goal of this paper is to propose a model, validated  with
respect to real field measurements in an operational network,
which can help answering  these questions.

Our objective  faces us with the  following important aspects of the 
modeling problem: 
(i) capturing the static but irregular and
heterogeneous network geometry, (ii) considering  the dynamic user service process 
at  individual network BS (cells), and last but not least 
(iii) taking into account the dependence between these service
processes. This latter dependence is  
due to the fact that the extra-cell interference makes the service of a
given cell depend on the ``activity'' of other cells in the network.
Historically, geometric  (i) and dynamic (ii) aspects are usually
addressed separately on the ground of stochastic geometry and queueing
theory, respectively.

Cellular network models based on the planar Poisson point process
have been shown recently to give tractable expressions for
many characteristics built from the powers of different BS received at one
given  location, as e.g. the signal-to-interference-and
noise ratio(s) (SINR) of the, so-called, typical user. They describe
potential resources of the network (peak bit-rates, spectral or energy
efficiency etc) but not yet its real performance when several  users
have to share these resources. On the other hand,
various classical queueing models can be tailored to represent the
dynamic resource sharing at one or several BS (as e.g. loss models for constant
bit-rates services and processor sharing queues for variable bit-rates
services).

Our model considered in this paper combines the stochastic-geometric
approach with the  queueing one to represent the network in its spacial
irregularity and temporal evolution.
It assumes the usual {\em multi-tier Poisson model for BS locations} with
shadowing and the {\em space-time Poisson process of call arrivals} independently marked by
data volumes.  Each station applies a processor sharing policy to serve
users which receive its signal as the strongest one, with the 
peak service rates depending on the respective SINR.
The mutual-dependence of cell performance (iii)
is captured  via a system of {\em cell-load (fixed point) equations}.
By the load we mean the ratio of the
actual traffic demand  to its critical value, which can be  interpreted, when 
it is smaller than one, as the busy probability in the classical
processor sharing queue.
The cell load equations make the load of  a given station  dependent
on the busy probabilities (hence loads) of other stations, by taking
them as weighting factors of the interference induced by these
stations. Given network realization, this decouples the temporal (processor-sharing) queueing processes  of
different cells, allowing us  to use the classical results to
evaluate their steady state characteristics (which depend on the
network geometry).
We (numerically) solve the cell-load  fixed point problem  calculating loads and
other characteristics of the individual cells. Appropriate spatial (network)
averaging of these characteristics, expressed  using the  formalism of
the typical cell offered by  Palm theory of point processes, 
provides useful macroscopic description of the network performance. 

The above approach is  validated by estimating the model parameters from the real field measurements 
of a given operational network and comparing the macroscopic network performance characteristics
calculated using this model to the performance of the real network.

{\em The remaining part of the paper} is organized as follows: 
In Section~\ref{ss.RelatedWork} we briefly present the related work. 
Our model is introduced in Section~\ref{s.ModelDescription}
and studied in Section~\ref{s.Analysis}. Numerical results 
validating our approach are  presented in Section~\ref{s.NumericalResults}.

\subsection{Related work}
\label{ss.RelatedWork}
There are several ``pure''  simulation tools developed for the performance  evaluation of 
cellular networks such as those
developed by the industrial contributors to 3GPP (\emph{3rd Generation
Partnership Project})~\cite{3GPP36814-900}, 
TelematicsLab LTE-Sim~\cite{Piro2011}, University of Vien LTE
simulator~\cite{Mehlfuhrer2011,Simko2012} and LENA
tool~\cite{Baldo2011,Baldo2012} of CTTC. 
They do not unnecessarily allow to identify the  macroscopic laws regarding network performance metrics.

A possible analytical approach to this problem is based on the information
theoretic %~\cite{CoverThomas1991}
characterization of the individual link performance; cf
e.g.~\cite{GoldsmithChua1997,Mogensen2007}, in conjunction with a queueing
theoretic %~\cite{Cohen1976} 
modeling and analysis of the user traffic
cf. e.g.~\cite{Borst2003,BonaldProutiere2003,HegdeAltman2003,BonaldBorstHegdeJP2009,RongElayoubiHaddada2011,KarrayJovanovic2013Load}%
. These works are usually focused on some particular aspects of the
network and do not consider a large, irregular, heterogeneous,  multi-cell scenario. 

Stochastic geometric
approach~\cite{HABDF:2009} to wireless communication networks consist in taking 
spatial averages over node (emitter, receiver) locations.
It was first shown in~\cite{ANDREWS2011} to give analytically tractable
expressions for the typical-user characteristics in Poisson 
models of cellular networks, 
with the Poisson assumption
being  justified  by representing highly irregular base station deployments in urban
areas~\cite{ChiaHan_etal2012}  or mimicking strong log-normal
shadowing~\cite{hextopoi,hextopoi-journal}, or  both. Expressions for the SINR coverage in 
multi-tier network models were
developed in~\cite{mukherjee2011downlink,DHILLON2012,mukherjee2012downlink,BlaszczyszynKK2013SINR}.
Several extensions of this initial model are reported in~\cite{mukherjee2014analytical}.
The concept of equivalence of heterogeneous networks (from the point
of view of its typical user), which we use in the present paper, was
recently formulated in~\cite{equivalence2013}, but previously  used 
in several works e.g. in~\cite{blaszczyszyn2010impact,MADBROWN2011,PINTO2012,BlaszczyszynKK2013SINR}.
%~\cite{MADBROWN2011,PINTO2012}.
 
The fixed-point cell-load equation was postulated independently in
~\cite{KarrayJovanovic2013Load} and~\cite{siomina2012analysis}
to capture the dependence of processor sharing queues modeling
performance of individual BS, in the context of regular hexagonal and
fixed deterministic network models, respectively.
Our present paper, combining stochastic geometry with queueing
theory complements~\cite{blaszczyszyn2014user}, where a homogeneous
network is considered, and~\cite{JovaQoS} where the distribution of the QoS
metrics in the heterogeneous network has been studied by simulation.
A network dimensioning  methodology based on this approach was recently proposed in~\cite{dimension}.

\section{Model description}
In this section we describe the components of our model.
\label{s.ModelDescription}

\subsection{Network geometry}
\subsubsection{Multi-tier network of BS} 
We consider a multi-tier cellular network consisting of $J$ types
(tiers) of BS characterized by different transmitting
powers $P_j$, $j=1,\ldots,J$. 
Locations of BS are modeled by independent homogeneous Poisson point processes
$\Phi_j$ on
the plane, of intensity $\lambda_j$ stations per $\mathrm{km}^2$.
Let $\Phi=\{X_n\}$ be the superposition of $\Phi_{1},\ldots,\Phi_J$
(capturing the locations of all BS of the network). 
 Denote by $Z_n\in\{1,\ldots,J\}$ the type of BS $X_n\in\Phi$
(i.e., the index of the tier it belongs to). 
It is known that $\Phi$ is a Poisson point process of intensity parameter $\lambda=\sum_{j=1}^J \lambda
_{j}$ and $Z_n$ form independent, identically distributed  (i.i.d) marks of
 $\Phi$ with $\Pro(Z_n=j)=\lambda_j/\lambda$.

\subsubsection{Propagation effects}
\label{sss.Propagation}
The propagation loss %due to the distance 
is modeled by a deterministic 
{\em path-loss function} $l(x)=(K\left\vert x\right\vert )^{\beta}$, where
$K>0$ and $\beta>2$\ are given constants, and some random propagation effects.
We split these effects into two categories conventionally called {\em
  (fast) fading} and {\em shadowing}. The former will be accounted in the
  model at the link-layer (in the peak bit-rate function
  cf.~(\ref{e.Shannon})). The latter impacts the choice of the serving
  BS and thus needs to be considered together with the network
  geometry. To this regard we assume that the shadowing between a given station $X_{n}\in\Phi$ and all
locations $y$ on the plane is modeled by some positive valued stochastic process
$\mathbf{S}_{n}\left(  y-X_{n}\right)  $. We assume that the processes $\mathbf{S}_{n}\left(
\cdot\right)  $ are i.i.d.  marks of
$\Phi$.~\footnote{The assumption that all types of base stations have
  the same distribution of the shadowing can be easily relaxed.}
Moreover we assume that $\mathbf{S}_{1}(y)$ are identically
distributed across $y$, but do not make  any assumption
regarding  the dependence of $\mathbf{S}_{n}(y)$ across $y$.

Thus the inverse of the power averaged over fast fading,  received at $y$ from BS $X_n$, denoted by
$L_{X_{n}}\left(  y\right)=L_{n}\left(  y\right)$ which we call (slightly abusing the
terminology) the propagation-loss from this station 
is given by
\begin{equation}
L_{X_n}\left(  y\right)  =\frac{l\left(  \left\vert y-X_{n}\right\vert
\right)  }{P_{Z_n}\mathbf{S}_n\left(  y-X_{n}%
\right)  }. \label{e.Propagation}%
\end{equation}
%(The above quantity depends not only on $X_{n}$\ but also on $Z_{n}$\ and $\mathbf{S}_{n}$. Indeed, it should be understood as $L_{X_{n}\left( \omega\right)  }\left(  y\right)  \left(  \omega\right)  $, but the elementary event $\omega$ is omitted as usual in probability theory.)
In what follows, we will often simplify the notation writing  $L_X(\cdot)$ for the propagation-loss of BS $X\in\Phi$.

\subsubsection{Service zones, SINR and peak bit-rates}
\label{sss.Cells}
We assume that each (potential) user located at $y$ on the plane is served by the BS offering the strongest received power among all
the BS in the network. Thus, the zone served by BS $X\in\Phi$,
denoted by $V(X)$,
which we keep calling  {\em cell} of $X$ (even if random shadowing
makes it need not to be a polygon or even a connected set) is given by
\begin{equation}
V\left(  X\right)  =\left\{  y\in\mathbb{R}^{2}:L_{X}\left(  y\right)  \leq
L_{Y}\left(  y\right)\;\text{for all\,} Y\in\Phi\right\}  \label{e.Cell}%
\end{equation}
We define the  (downlink) SINR at location $y\in V\left(  X\right)$
(with respect to the serving  BS $X\in\Phi$) as follows
\begin{equation}
\SINR\left(  y,\Phi\right) :=\frac{1/L_{X}\left(  y\right)}%
{N+\sum_{Y\in\Phi\backslash\left\{  X\right\}  } \varphi_Y/L_{Y}\left(
y\right)  },\label{e.SINR}%
\end{equation}
where $N$\ is the noise power and the {\em activity factors} $\varphi_Y\in[0,1]$ account (in a way that
will be made specific in Section~\ref{ss.fixedPoing}) for the activity of stations
$Y\in\Phi$. In general, we assume that 
$\varphi_Y$ are additional (not necessarily independent) marks of the
point process $\Phi$, possibly dependent on tiers and
shadowing of all BS.

We assume that the {\em (peak) bit-rate} at location $y$, defined as
the number of bits per second a user located at $y$ can download 
when served alone by its BS, is some function $R(\SINR)$ of the SINR. Our general analysis presented in Section~\ref{s.Analysis} does not
depend on any particular form of this function. A specific 
expression will be assumed for the numerical results in Section~\ref{s.NumericalResults}.

\subsection{Network users}
\subsubsection{User-arrival process}
\label{sss.arrival}
We consider variable bit-rate (VBR) traffic; i.e., users arrive to the network
and require to transmit some volumes of data at bit-rates induced  by the
network. We assume a {\em homogeneous time-space Poisson point process
  of user arrivals} of intensity $\gamma$ arrivals per second per
  $\mathrm{km}^2$. This means that the time  between two successive
arrivals in a given zone of surface $S$ is exponentially distributed
with  parameter $\gamma\times S$, and all users arriving
to this zone take their locations independently and uniformly.
The time-space process of user arrivals is independently marked 
by random, identically distributed  volumes of data the users want to
download from their respective serving BS. These volumes are arbitrarily
distributed and have mean $1/\mu$ bits. 

The above arrival process
induces the {\em traffic demand per surface unit}
$\rho=\gamma/\mu$
expressed in bits per second per $\mathrm{km}^{2}$.
The {\em traffic demand in the cell  of BS $X\in\Phi$} equals
\begin{equation}
\rho\left(  X\right)  =\rho\left\vert V\left(  X\right)  \right\vert, \label{e.TrafficDemand}%
\end{equation}
where $\left\vert A\right\vert $\ denotes the surface of the set $A$;
$\rho(X)$  is expressed in bits per second.

\subsubsection{Processor-sharing service policy}
We shall assume that the BS allocates an equal fraction of its
resources to all users it serves at a given time. Thus, when there are $k$ users in a cell,
each user obtains a bit-rate equal to its  peak bit-rate divided by $k$. 
More explicitly, if a base station
located at $X$ serves $k$ users located at $y_{1},y_{2},\ldots,y_{k}\in
V\left(  X\right)  $\ then the bit-rates of these users are equal to
$R\left(  \mathrm{SINR}\left(  y_{j},\Phi\right)  \right)/k  $,
$j= 1,2,\ldots,k  $, respectively.
Users having completed their service (download of the requested
volumes) leave  the system. 

\subsection{Time-averaged cell characteristics}
\label{s.CellCharacteristics}
Given the network realization (including the shadowing and the cell activity
factors), the performance of each cell $V(X)$ of $X\in\Phi$ corresponds to 
a (spatial version of the) processor sharing-queue. More specifically,
due to complete independence  property of the Poisson process of
arrivals, the temporal dynamics of these queses are independent. Thus,
we can use the classical queuing-theoretic results regarding
processor-sharing queues to describe the time-averaged (steady-state)
characteristic of all individual  cells. 
Besides the traffic demand $\rho(X)$ already specified  in Section~\ref{sss.arrival},
these characteristics are: the  critical traffic $\rho_c(X)$, cell load $\theta(X)$, mean number of
users $N(X)$, average user throughout $r(X)$, busy (non-idling)
probability $p(X)$. 
In what follows we present these characteristics in a form tailored to
our wireless context; cf~\cite{BonaldProutiere2003}.
All these characteristics can be seen as further, general
(non-independent) marks of the point process $\Phi$ and depend also on
BS types, their activity factors  and shadowing processes.

\subsubsection{Critical traffic}
The processor-sharing queue of the base station $X\in\Phi$ is stable
if and only if its traffic demand $\rho(X)$ is smaller than
the critical value which is the 
harmonic mean of the peak bit-rates over the cell; cf~\cite{KarrayJovanovic2013Load}
\begin{equation}
\rho_{\mathrm{c}}\left(  X\right)  :=\left\vert V\left(  X\right)  \right\vert
\left(  \int_{V\left(  X\right)  }R^{-1}\left(  \mathrm{SINR}\left(
y,\Phi\right)  \right)  dy\right)^{-1}\,. \label{e.CriticalTraffic}%
\end{equation}

\subsubsection{Cell load} We define it as the ratio between
the (actual) cell traffic demand and its critical value
\begin{equation}
\theta\left(  X\right): =\frac{\rho\left(  X\right)  }{\rho_{\mathrm{c}%
}\left(  X\right)  }=\int_{V\left(  X\right)  }\rho R^{-1}\left(
\mathrm{SINR}\left(  y,\Phi\right)  \right)  dy\,. \label{e.Load1}%
\end{equation}

\subsubsection{Mean number of users}
The mean number of users in the steady state of the processor sharing
queue at BS $X\in\Phi$ can be expressed as 
\begin{equation}\label{e.UsersNumber1}%
N\left(  X\right):=
\begin{cases}\displaystyle
\frac{\theta\left(  X\right)  }{1-\theta\left(  X\right)  }& 
\text{if $\theta(X)<1$}\\
\infty&\text{otherwise}\,.
\end{cases}
\end{equation}

\subsubsection{User throughput}
\label{sss.throughput-cell} is defined as the ratio between the
mean  volume request $1/\mu$ and the mean typical-user service time in
the cell $X$. By the Little's law it can be expressed as
\begin{equation}\label{e.UserThroughput1}%
r\left(  X\right):=\frac{\rho(X)}{N(X)}\,.
\end{equation}

\subsubsection{Busy probability} The probability that the BS
$X\in\Phi$ is not idling (serves at least one user) in the steady
state is equal to
\begin{equation}
p\left(  X\right)  =\min\left(  \theta\left(  X\right)  ,1\right)\,.
\label{e.Proba}%
\end{equation}
It is easy to see that all the above characteristics (marks) of the BS
$X\in\Phi$ can be expressed using the traffic demand $\rho(X)$ and
the cell load $\theta(X)$ in the following order
\begin{align}\label{e.CriticalTraffic2}
\rho_c(X)&=\frac{\rho(X)}{\theta(X)}\,,\\
r(X)&=\max(\rho_{\mathrm{c}}\left(  X\right)  -\rho\left(
X\right)  ,0) \label{e.UserThroughput}\,,\\%
N\left(  X\right)  &=\frac{\rho\left(  X\right)  }{r\left(  X\right)  }\,.
\label{e.UsersNumber}%
\end{align}

%\begin{equation}\label{e.r-theta}
%r(X)=
%\begin{cases}\displaystyle
%\rho(X)\Bigl(\frac{1}{\theta(X)}-1\Bigr)&\text{if  $\theta(X)<1$}\\
%0&\text{otherwise}
%\end{cases}
%\end{equation}

\subsection{Spatial inter-dependence of cells --- cell load equations}
\label{ss.fixedPoing}
The individual cell characteristics described in the previous section depend on
the location of all base stations, shadowing realizations but also on
the cell activity factors $\varphi_X$, $X\in\Phi$, introduced in
Section~\ref{sss.Cells} to weight the extra cell interference in the
SINR expression, and which have been arbitrary numbers between
0 and 1 up to now. These factors suppose to account for the fact that 
BS might not transmit with their respective maximal powers $P_j$ depending
on the BS types $j=1,\ldots,J$ all the time. 

It is quite natural to think that BS transmit only when they serve at least one
user.~\footnote{Analysis of more sophisticated power control schemes is beyond
the scope of this paper.} Taking this fact into account in an exact way
requires introducing in the denominator of~(\ref{e.SINR}) the indicators that
a given station $Y\in\Phi$ at a given time is not idling. This, in
consequence, would lead to the probabilistic dependence of the service process
at different cell, thus revoking  the explicit expressions for their
characteristics presented in Section~\ref{s.CellCharacteristics} 
and the model becomes non-tractable.~\footnote{We are even not aware
  of any result regarding the 
stability of such a family of dependent queues.}
For this reason, we take into account whether $Y$ is idling or not in a
simpler way, multiplying its maximal transmitted power  by the  \emph{probability} $p(Y)$ that
it is busy in the steady state. 
In other words, in the SINR expression~(\ref{e.SINR}) we take 
$\varphi_Y=p(Y)$
where $p(Y)$ is  given by~(\ref{e.Proba}); i.e., 
\begin{equation}
\mathrm{SINR}\left(  y,\Phi\right)  =\frac{\frac{1}{L_{X}\left(  y\right)  }%
}{N+\sum_{Y\in\Phi\backslash\left\{  X\right\}  }\frac{\min\left(
\theta\left(  Y\right)  ,1\right)  }{L_{Y}\left(  y\right)  }}\,.
\label{e.LoadInterference}%
\end{equation} 
We call this
model {\em (load)-weighted interference model}. Clearly this assumption means
that $\theta(X)$ cannot be calculated independently for all cells but
rather are solutions of the following fixed point problem, which we
call {\em cell load equations} 

\begin{equation}
\theta\left(  X\right)  =\rho\int_{V\left(  X\right)  }R^{-1}\left(
\frac{\frac{1}{L_{X}\left(  y\right)  }}{N+\sum_{Y\in\Phi\backslash\left\{
X\right\}  }\frac{\min\left(  \theta\left(  Y\right)  ,1\right)  }%
{L_{Y}\left(  y\right)  }}\right)  dy\,. \label{e.FixedPoint}%
\end{equation}
This is a system of equations which needs to be solved for $\left\{  \theta\left(
X\right)  \right\}  _{X\in\Phi}$ given network and shadowing
realization. In the remaining part of this paper we  assume
that such a solution exists and is unique.~\footnote{Note that the mapping in the right-hand-side of~(\ref{e.FixedPoint})
is increasing  in all $\theta(Y)$, $Y\in\Phi$ provided function $R$ is increasing.
Using this property it is easy to see that successive iterations of this mapping
started off $\theta(Y)\equiv 0$ on one hand side and off $\theta(Y)=1$
(full interference model) on the other side,
converge to a minimal and maximal solution of~(\ref{e.FixedPoint}),
respectively. The uniqueness of the solution (in the Poisson or more
general) network is an interesting theoretical question, which is however beyond the scope of this paper. 
A very similar problem (with finite number of stations and a discrete traffic demand)
is considered in~\cite{siomina2012analysis}, where the uniqueness of
the solution is proved.}
The other characteristics of each cell are then deduced from the cell
load and traffic demands using the relations described in
Section~\ref{s.CellCharacteristics}.

%Let us finally mention that for comparison we will also consider 
%a \emph{full interference model}, where the interfering BS are assumed to
%transmit continuously, which consists in taking $\varphi_Y=1$ for all $Y\in\Phi$.
%We will see in the numerical section that the weighted-interference
%model  fits better to real field
%measurements than the full interference model. 
%In particular, the fact that
%the average load of each category of BS\ increases with its power through a
%simple law will be validated by real field measurements in the numerical section.

\section{Model analysis}
\label{s.Analysis}
We begin our analysis by recalling some useful results regarding the Poisson
network model. Next, in Section~\ref{ss.TypicalCell} we present our
main results and in Section~\ref{ss.MeanCell} postulate some
simplified approach inspired by these results.
 
\subsection{Preliminaries: typical and zero-cell of the multi-tier network}
We briefly recall here the notions of the typical and zero-cell,
usually considered for the Voronoi tessellation and here regarding
our network cells.
Both objects will play their respective roles in the remaining part of
the paper.  

We denote by $\Pro$ the
probability corresponding to the stationary distribution of our model 
as described in Section~\ref{s.ModelDescription}.

\subsubsection{The typical cell}
\label{e.TypCell}
This is a mathematical formalization of a cell whose BS is ``arbitrarily chosen''
from the set of all stations, without any bias towards its
characteristics,  in particular its type and the cell size.  The formalization is made
on  the ground of Palm theory, where the typical cell $V(0)$ is
this of the  BS $X_0=0$ located at the origin under the {\em Palm
probability} $\Pro^0$. By the Slivnyak's theorem the Palm distribution of 
the Poisson process corresponds to the homogeneous (stationary) one,
with the  ``extra'' point $X_0=0$ added at the origin. In the case of
i.i.d. marked Poisson process, as in our case,  
this extra point gets and independent copy of the   mark, with
the original mark distribution.

Note that in our network the probability that an ``arbitrarily
chosen'' BS is of type $j$, $j=1,\ldots,J$, is equal to
$\lambda_j/\lambda$. More formally, 
\begin{equation}\label{e.PrZ0}
\Pro^0(Z_0=j)=\lambda_j/\lambda\,.
\end{equation}

We remark, that   the typical cell does not
have any physical existence in a given network.  It is rather a useful
mathematical tool, in the sense that 
%its  mean  characteristics (under $\Pro^0$)
%correspond to the empirical averages  of over  many cells $V(X)$,
%$X\in\Phi$, considered under probability $\Pro$ corresponding to the
%stationary distribution of the model.}  
the mathematical expectations  under
$\Ex^0$ of the typical cell $V(0)$ characteristics (as the cell
traffic demand $\rho(0)$, cell load $\theta(0)$, etc) can be interpreted
as network-averages of the (already time-averaged) cell performance metrics. 
For example the network-averaged traffic demand per cell, considering
all cells or only cells  of type $j=1,\ldots,J$, equal,
 respectively
\begin{align}\label{e.TypicalTraffic-general}%
\bar{\rho}&:=\Ex^0[\rho(0)]=\lim_{\left\vert A\right\vert \rightarrow\infty}\frac{1}%
{\Phi\left(  A\right)  }\sum_{X\in\Phi\cap A}\rho\left(  X\right)\,,\\
\bar{\rho}_{j}&:=\Ex^0[\rho(0)\,|\,Z_0=j]=\lim_{\left\vert A\right\vert \rightarrow\infty}\frac{1}%
{\Phi_{j}\left(  A\right)  }\sum_{X\in\Phi_{j}\cap A}\rho\left(  X\right)\,.
\label{e.TypicalTraffic}%
\end{align}
where $A$ denotes a disc centered at the origin, of radius
increasing  to infinity. The convergence is $\Pro$-almost sure and follows
from the ergodic theorem for point processes (see~\cite[Theorem~13.4.III]{DaleyVereJones2003}).
%\cite[Theorem~4.2.1]{BaccelliBlaszczyszyn2009T1}
%Moreoer the networ-averaged traffic demand per cell of type $j$ equals
%\begin{equation}
%\bar{\rho}_{j}:=\Ex^0[\rho(0)\,|\,Z_0=j]=\lim_{\left\vert A\right\vert \rightarrow\infty}\frac{1}%
%{\Phi_{j}\left(  A\right)  }\sum_{X\in\Phi_{j}\cap A}\rho\left(  X\right)\,.
%\label{e.TypicalTraffic}%
%\end{equation}
We define similarly the network-average load (overall and per cell type)
\begin{align}
\bar{\theta}&:=\Ex^0[\theta(0)]\,,\\
\bar{\theta}_{j}&:=\Ex^0[\theta(0)\,|\,Z_0=j]\quad j=1,\ldots,J.
\end{align}
The convergence analogue to~(\ref{e.TypicalTraffic-general}), (\ref{e.TypicalTraffic})
holds for each of the previously considered local characteristics.
%$\mathbf{E}^{0}[\rho_{c}(0)]$,
%$\mathbf{E}^{0}[r(0)]$, $\mathbf{E}^{0}[N(0)]$, $\mathbf{E}^{0}[p(0)]$ and
%$\mathbf{E}^{0}[\theta(0)]$. The convergence is $\Pr$~almost sure and follows
%from the ergodic theorem for point processes (see~\cite[Theorem~4.2.1]%
%{BaccelliBlaszczyszyn2009T1},~\cite[Theorem~13.4.III]{DaleyVereJones2003}).
%for more details, in
%particular for a more general form of increasing window $A$.
%and similarly for $\rho_c(X)$, $r(X)$, $N(X)$,
%$p(X)$ and $\theta(X)$ introduced in Section~\ref{s.QoS}.
However (at least for Poisson network) it is not customary to consider directly   
% as we will explain in what follows, 
%\emph{not all} of these
%\emph{mean-typical cell} 
$\Ex^0[N(0)]$  since  the (almost sure) existence of some (even arbitrarily small) fraction of BS
$X$ which are not stable (with $\rho(X)\ge\rho_{c}(X)$, hence $N(X)=\infty$)
makes $\Ex^{0}[N(0)]=\infty$.~\footnote{%
For a well dimensioned network one does not expect
unstable cells. For a perfect hexagonal network model $\Phi$ without
shadowing \emph{all} cells
are stable or unstable depending on the value of the per-surface traffic
demand $\rho$. For an (infinite) homogeneous Poisson model $\Phi$,
for arbitrarily small $\rho$ there exists a non-zero fraction of BS
$X\in\Phi$, which are non-stable. This fraction is very small for reasonable
$\rho$, allowing to use Poisson to study QoS metrics which, unlike
$\mathbf{E}^{0}[N(0)]$, are not ``sensitive'' to this artifact.
}
 
Also, as we will explain in what follows,  $\Ex^0[r(0)]$ does not
have a natural interpretation. In particular it {\em cannot} be interpreted
as the mean user throughput.

\subsubsection{Zero cell} This is the  cell (of the stationary
distributed network) that covers the origin $0$
of the plane, which plays the role of an arbitrarily fixed location.
The characteristics of the zero-cell correspond to the characteristics
of the cell which serves the typical user. Clearly this is a size-biased choice and indeed the zero cell has
different distributional characteristics from the typical cell.
Let us denote by $X^*$ the location of the BS serving the zero-cell
and its type by $Z^*$.

We will recall now a useful result regarding
multi-tier networks, from which we will derive the distribution of
$Z^*$; cf~\cite[Lemma~1]{equivalence2013}.
\begin{lemma}\label{l.Lambda}
Assume that $\mathbb{E}\left[  S^{2/\beta}\right]  <\infty$. Then $\hat{\Phi
}=\left\{  (L_n=L_{n}(0),Z_n)\right\}  _{n}$ is a Poisson point process on
$[0,\infty)\times\{1,\ldots,J\}$ with intensity measure
\begin{equation}\label{e.Lambda}
\Lambda\left((  0,t\right]\times\{j\}):=\Ex[\#\{n:L_n\le t, Z_n=j\}]  =a_jt^{2/\beta}\,,%
\end{equation}
$t\ge 0$, $j=1,\ldots,J$, where
\begin{equation}\label{e.aj}
a_{j}:=\frac{\pi\mathbb{E}\left[  S^{2/\beta}\right]  }{K^{2}}\lambda_{j}%
P_{j}^{2/\beta}\,.
\end{equation}
\end{lemma}
\begin{remark}\label{r.Z*}
The form~(\ref{e.Lambda}) of the intensity measure $\Lambda$ of
$\hat{\Phi}$ allows us
to conclude that the point process $\{L_n(0)\}_n$ of propagation-loss values (between all base stations and the
origin) is a Poisson point process of intensity
$\Lambda((0,t]\times\{1,\ldots,J\})=at^{2/\beta}$, where
%$a:=\sum_{j=1}^Ja_j$.
\begin{equation}\label{e.a}
a:=\sum_{j=1}^Ja_j\,. %=\frac{\pi\mathbb{E}\left[  S^{2/\beta}\right]  }{K^{2}}\sum_{j=1,\ldots, J}%
%\lambda_{j}P_{j}^{2/\beta}\,.%
\end{equation}
Moreover, the types $Z_n$ of the BS corresponding to the respective
propagation-loss values $L_n$ constitute i.i.d. marking of this latter
process of propagation-loss values, 
with the probability that an arbitrarily chosen propagation-loss
comes from a BS of type $j$ having  probability $a_j/a$.
In particular, for the serving station (offering the smallest propagation-loss)
we have
\begin{equation}\label{e.PrZ*}
\Pro\{\,Z^*=j\,\}=a_j/a\,.\footnote{Interpreting~(\ref{e.PrZ0}) and  (\ref{e.PrZ*}) we can say
   that an arbitrarily chosen BS is of type $j$ with probability
   $\lambda_j/\lambda$, while an arbitrarily chosen propagation-loss
   (measured at the origin) comes from a BS of type $j$ with
   probability $a_j/a$.}
\end{equation}
\end{remark}
Our second  remark on the result of Lemma~\ref{l.Lambda} regards an
equivalent way of generating the Poisson point process of
intensity~(\ref{e.Lambda}).

\begin{remark}\label{r.equivalence}
Consider a homogeneous Poisson network of intensity $\lambda$, in
which all stations emit with the same power
\begin{equation}
P=\left(  \sum_{j=1}^ J\frac{\lambda_{j}}{\lambda}P_{j}^{2/\beta}\right)
^{\beta/2}\, \label{e.Power}%
\end{equation}
and assume the same model of the propagation-loss with shadowing as described
in Section~\ref{sss.Propagation}. Let us ``artificially'' (without
altering the power $P$) mark these
BS by randomly, independently selecting a mark  $j=1\ldots,J$ for each
station with probability $a_j/a$. A direct
calculation shows  that the marked propagation-loss process observed
in this homogeneous network by a user located at the origin, analogue 
to $\hat\Phi$, has the same intensity measure $\Lambda$ given
by~(\ref{e.Lambda}). Consequently, the distribution of all user/network
characteristics, which are functionals of the marked propagation-loss
process $\hat\Phi$ can be equivalently calculated using this {\em equivalent homogeneous}
model. %We will use this observation in~Section~\ref{s.MeanCell}. 
\end{remark}

\subsection{Global network performance metrics}
\label{ss.TypicalCell}
The objective of this section is to express pertinent,
global network characteristics
%   consider network-averages of cell characteristics
%described in Section~\ref{s.CellCharacteristics}, mathematically
%captured  by the Palm expectations of the typical cell, 
and relate them 
to  mean throughput of the typical user of the network.
%\emph{typical cell} for each category of BS. 

\subsubsection{Traffic and load per cell}
The mean traffic demand and load of the
typical cell, globally and per cell type, can be expressed as follows.
\begin{proposition}
\label{p.TypicalDemand}
 We have for the traffic demand
\begin{align}
\bar{\rho}& =\frac{\rho}{\lambda}\nonumber\\
\bar{\rho}_{j}  &
=\bar{\rho}\frac{P_{j}^{2/\beta}}{P^{2/\beta}},\quad j=1,\ldots,J,
%\frac{\rho a_{j}}{\lambda_{j}a},\quad j=1,\ldots, J
\label{e.TrafficPerType}%
\end{align}
where $P$\ is the ``equivalent network'' power given  by~(\ref{e.Power}).
\end{proposition}

\begin{proof}
%Since the point process $\Phi$\ is ergodic, it follows from the ergodic
%theorem for point processes that~\cite[Proposition~12.2.VI]%
%{DaleyVereJones1988}
We have 
\[
\bar{\rho}=\Ex^{0}\left[  \rho\left(  0\right)  \right]  =\rho
\Ex^{0}\left[  \left\vert V\left(  0\right)  \right\vert \right]
=\frac{\rho}{\lambda}\,,%
\]
where the second equality is due to~(\ref{e.TrafficDemand}) and the last
one follows from the inverse formula of Palm calculus~\cite[Theorem
4.2.1]{BaccelliBlaszczyszyn2009T1} (which may be extended to the case where
the cell associated to each BS is not necessarily the Voronoi cell; the only
requirement is that the user located at $0$ belongs to a unique cell almost
surely). Similarly,
\begin{align*}
\bar{\rho}_{j}  
% =\lim_{\left\vert A\right\vert \rightarrow\infty}\frac
%{\Phi\left(  A\right)  }{\Phi_{j}\left(  A\right)  }\frac{1}{\Phi\left(
%A\right)  }\sum_{X\in\Phi\cap A}\rho\left(  X\right)  \ind\left\{  X\in\Phi
%_{j}\right\} \nonumber\\
%&  =\frac{\lambda}{\lambda_{j}}\Ex^{0}\left[  \rho\left(  0\right)
%\ind left\{  0\in\Phi_{j}\right\}  \right] \nonumber\\
%&  =\Ex^{0}\left[  \rho\left(  0\right)  |0\in\Phi_{j}\right]\\
&=\rho\Ex^{0}\left[  \left\vert V\left(  0\right)\,  \right\vert
  \,Z_0=j\right]  %\label{e.TrafficPerType1}
 =\rho\frac{\Ex^{0}\left[  \left\vert V\left(  0\right)  \right\vert
\times\ind\left\{  Z_0=j\right\}  \right]  }{\mathbb{P}^{0}\left(
Z_0=j\right)  }\nonumber\\
&  =\frac{\rho}{\lambda}\frac{\mathbb{P}\left(  Z^{\ast}=j\right)
}{\mathbb{P}^{0}\left( Z_0=j\right)  }\nonumber
=\bar\rho\frac{a_{j}/a}{\lambda_{j}/\lambda}
=\bar\rho \frac{P_{j}^{2/\beta}}{P^{2/\beta}}\,,%\label{e.Surface}%
\end{align*}
where  the third equality follows again from the inverse formula of
Palm calculus and the two remaining ones from~(\ref{e.PrZ0}), (\ref{e.PrZ*}) and
(\ref{e.aj}), (\ref{e.a}), respectively.
\end{proof}

\begin{proposition}
\label{p.TypicalLoad} We have for the cell load
\begin{align}
\bar{\theta}  &  =\frac{\rho}{\lambda}\Ex\left[  R^{-1}\left(  \mathrm{SINR}%
\left(  0,\Phi\right)  \right)  \right]\,, \label{e.Load}\\
\bar{\theta}_{j}  &  %=\bar{\theta}\frac{\lambda a_{j}}{\lambda_{j}a}%
=\bar{\theta}\frac{P_{j}^{2/\beta}}{P^{2/\beta}},\quad j=1,\ldots, J\,,
\label{e.LoadPerType}%
\end{align}
where $P$\ is given by~(\ref{e.Power}).
\end{proposition}

\begin{proof}
Denote $g\left(  y,\Phi\right)  =R^{-1}\left(
\mathrm{SINR}\left(  y,\Phi \right)  \right)$. 
In the same lines as the proof of Proposition~\ref{p.TypicalDemand},
by the inverse formula of Palm calculus 
%\begin{align*}
$\bar{\theta}   =\Ex^{0}\left[  \theta\left(  0\right)  \right] 
 =\frac{\rho}{\lambda
}\Ex\left[  g\left(  0,\Phi\right)  \right]$.
Similarly
\begin{align*}
\bar{\theta}_{j}&=\Ex^{0}\left[  \theta\left(  0\right)  \,|\,Z_0=j\right]\\
%&  =\Ex^{0}\left[  \theta\left(  0\right)  \ind\left\{ Z_0=j\right\} \right] 
% /\Pro^{0}\left(  Z_0=j\right) \\
&  =\rho\Ex^{0}\left[  \int_{V\left(  0\right)  }g\left(  y,\Phi\right)
\ind\left\{ Z_0=j\right\}  dy\right]  /\mathbb{P}^{0}\left(
Z_0=j\right) \\
&  =\frac{\rho}{\lambda}\Ex\left[  g\left(  0,\Phi\right)  \ind\left\{
Z^{\ast}=j\right\}  \right]  /\mathbb{P}^{0}\left(  Z_0=j\right) \\
%&  =\frac{\rho}{\lambda}\Ex\left[  g\left(  0,\Phi\right)  \right]
%\frac{\Pro\left( Z^*=j\right)  }{\mathbb{P}^{0}\left(
%Z_0=j\right)  }\\
&=\frac{\rho}{\lambda}\Ex\left[  g\left(
0,\Phi\right)  \right]  \frac{a_{j}/a}{\lambda_{j}/\lambda}\,,%
\end{align*}
where the third  equality follows from the inverse formula of Palm
calculus, and the fourth equality from the independent marking of the
propagation-loss process by the BS types; cf Remark~\ref{r.Z*}.
\end{proof}

\subsubsection{Number of users per cell and mean user throughput}
For the reasons already explained at the end of
Section~\ref{e.TypCell} it is more convenient to average the number of
users per cell in the stable part of the network. To this regard we define
the  network-averaged number of users per {\em stable} cell as 
\begin{align*}
\bar{N} %:=\lim_{\left\vert A\right\vert \rightarrow\infty}\frac{1}%
%{\Phi\left(  A\right)  }\sum_{X\in\Phi\cap A}N\left(  X\right)  \ind\left\{
%\theta\left(  X\right)  <1\right\} \\
 :=\Ex^{0}\left[  N\left(  0\right)  \ind\left\{  \theta\left(  0\right)
<1\right\}  \right]
\end{align*}
and similarly for each cell tier  $j=1,\ldots,J$\
\begin{align}
\bar{N}_{j} % &  :=\lim_{\left\vert A\right\vert \rightarrow\infty}\frac
%{1}{\Phi_{j}\left(  A\right)  }\sum_{X\in\Phi_{j}\cap A}N\left(  X\right)
%1\left\{  \theta\left(  X\right)  <1\right\} \label{e.TypicalUsersNumber}\\
: =\Ex^{0}\left[  N\left(  0\right) \ind\left\{  \theta\left(  0\right)
<1\right\}  |Z_0=j\right]\,. \nonumber
\end{align}

Note that the mean traffic demand $\bar\rho$, load $\bar\theta$ and number
of users $\barN$ per (stable) cell characterize network performance from the
point of view of its typical (or averaged) cell. We move now to a
typical user performance metric that is its mean throughput. This
latter QoS metric is traditionally (in queueing theory) defined  as 
the mean data volume requested by the typical user to the mean
service duration of the typical user.  In what follows we apply this
definition (already retained at the local, cell level in
Section~\ref{sss.throughput-cell}) globally to the whole network, 
filtering out the impact of unstable cells.

Denote by $\mathcal{S}_{j}$\ the union of stable cells of type $j=1,\ldots, J$; that
is
%\[
$\mathcal{S}_{j}=\bigcup_{X\in\Phi_{j}:\theta\left(  X\right)  <1}V\left(
X\right)$ 
%\]
and $\mathcal{S}=\bigcup_{j=1}^J\mathcal{S}_{j}$. 
Let  $\pi^\calS$ ($\pi^\calS_j$) be the probability the typical
user is served in a stable cell (of type $j=1,\ldots,J$) 
\begin{align*}
\pi^\calS &  =\mathbb{P}\left(  \theta\left(  X^{\ast}\right)  <1\right)\\
\pi^\calS_{j}  &  =\mathbb{P}\left(  \theta\left(  X^{\ast}\right)  <1\,|\,Z^{\ast}=j\right)  ,\quad j=1,\ldots, J\,,
\end{align*}
where (recall) $X^{\ast}$ is the BS whose cell covers the origin and
$Z^*$ is its type.
Note that
$\pi^{\mathcal{S}}_j=a/a_j\mathbf{E}[\mathbf{1}\{0\in\mathcal{S}_j\}]$ and
thus it can be related to the {\em volume fraction} of the stable
  part of the network served by tier  $j$ and similarly for
  $\pi^\calS=\mathbf{E}[\ind\{0\in\mathcal{S}\}]$.

We define the  {\em (global) mean user throughput} as 
\[
\bar{r}:=\lim_{\left\vert A\right\vert \rightarrow\infty}\frac{1/\mu
}{\text{mean call duration in }A\cap\mathcal{S}}%
\]
and for each cell type  $j=1,\ldots,J$,
\[
\bar{r}_{j}:=\lim_{\left\vert A\right\vert \rightarrow\infty}\frac{1/\mu
}{\text{mean call duration in }A\cap\mathcal{S}_{j}}\,,%
\]
where $A$ denotes a disc centered at the origin of radius
increasing  to infinity. These limits exist almost surely by the
ergodic theorem; cf~\cite[Theorem~13.4.III]{DaleyVereJones2003}.
Here is our main result regarding this mean user QoS. It can be seen
as a consequence of a spatial version of the Little's law.
\begin{proposition}
\label{p.TypicalThroughput}We have for the mean user throughput 
\begin{align}
\bar{r}  &  =\frac{\bar{\rho}}{\bar{N}}\pi^\calS\nonumber\\
\bar{r}_{j}  &  =\frac{\bar{\rho}_{j}}{\bar{N}_{j}}\pi^\calS_{j},\quad j=1,\ldots, J
\label{e.TypicalThroughput}%
\end{align}
\end{proposition}

\begin{proof}
Let $W_{j}=\bigcup_{X\in A\cap\mathcal{S}_{j}}V\left(  X\right)  $. Consider
call arrivals and departures to $W_{j}$. 
By Little's law
\[
N^{W_{j}}=\gamma\left\vert W_{j}\right\vert T^{W_{j}}\,,%
\]
where $T^{W_{j}}$ is the mean call duration in $W_{j}$
and $N^{W_{j}}$ is the steady-state mean number of users in $W_{j}$. 
Thus mean user throughput, with users restricted to  $W_{j}$, equals
\begin{align*}
\frac{1/\mu}{T^{W_{j}}}  &  =\frac{\rho\left\vert W_{j}\right\vert }{N^{W_{j}%
}}
%&  =\rho\frac{\sum_{X\in A\cap\mathcal{S}_{j}}\left\vert V\left(  X\right)
%\right\vert }{\sum_{X\in A\cap\mathcal{S}_{j}}N\left(  X\right)  }\\
  =\rho\frac{\sum_{X\in A\cap\Phi}\left\vert V\left(  X\right)  \right\vert
\ind\left\{  \theta\left(  X\right)  <1,X\in\Phi_{j}\right\}  }{\sum_{X\in
A\cap\Phi}N\left(  X\right)  \ind\left\{  \theta\left(  X\right)  <1,X\in\Phi
_{j}\right\}  }\,.%
\end{align*}
Letting $\left\vert A\right\vert \rightarrow\infty$, it follows from the ergodic
theorem that
\[
\bar{r}_{j}=\rho\frac{\Ex^{0}\left[  \left\vert V\left(  0\right)
\right\vert \ind \left\{  \theta\left(  0\right)  <1,Z_0=j\right\}
\right]  }{\Ex^{0}\left[  N\left(  0\right)  \ind \left\{  \theta\left(
0\right)  <1,Z_0=j\right\}  \right]  }\,.%
\]
By the inverse formula of Palm calculus 
\[
\Ex^{0}\left[  \left\vert V\left(  0\right)  \right\vert \ind \left\{
\theta\left(  0\right)  <1,Z_0=j\right\}  \right]  =\frac{1}{\lambda
}\mathbb{P}\left(  \theta\left(  X^{\ast}\right)  <1,Z^{\ast}=j\right)
\]
and consequently
\begin{align*}
\bar{r}_{j}  & =\frac{\rho}{\lambda}\frac{\mathbb{P}\left(  \theta\left(
X^{\ast}\right)  <1,Z^{\ast}=j\right)  }{\mathbb{P}^{0}\left(
Z_0=j\right)  \bar{N}_{j}}\\
%&  =\frac{\rho}{\lambda}\frac{\mathbb{P}\left(  Z^{\ast}=j\right)
%\mathbb{P}\left(  \theta\left(  X^{\ast}\right)  <1|Z^{\ast}=j\right)
%}%
%{\mathbb{P}^{0}\left(Z_0=j\right)  \bar{N}_{j}}\\
&=\frac{\bar{\rho}_{j}}{\bar{N}_{j}}\mathbb{P}\left(  \theta\left(  X^{\ast
}\right)  <1|Z^{\ast}=j\right)
=\frac{\bar{\rho}_{j}}{\bar{N}_{j}}\pi_{j}\,.%
\end{align*}
The expression for $\bar{r}$\ may be proved in the same lines as above.
\end{proof}

The mean  number of users per stable cell  and  the
mean user  throughput in the stable part of the network do not
admit explicit analytic expressions. We calculate these expressions
by Monte-Carlo simulation of the respective expectations with respect
to the distribution of the Poisson network model. We call this
semi-analytic approach the  {\em typical cell} approach.

\subsection{Mean cell approach}
\label{ss.MeanCell}
We will propose now a  more heuristic approach, in which we try to capture
the performance of the heterogeneous network considering $J$ simple
M/G/1 processor sharing queues related to each other via their cell
loads, which solve a simplified version of the cell load equation.

Recall  that in the original  approach, in the cell load fixed point equation~(\ref{e.FixedPoint})
we have an unknown cell loads $\theta(X)$ for each cell of the
network.  Recall
also that knowing all these cell loads and the cell traffic demands
(which depend directly on the cell surfaces)  we can
calculate all other cell characteristics. We will consider now a
simpler ``mean'' cell load fixed point equation in which all cells of
a given type $j=1,\ldots,J$  share the same constant
unknown $\tilde\theta_j$.~\footnote{Recall that~(\ref{e.FixedPoint}) is already a
  simplification of the reality in which the extra cell interference
  should be wighted by the dynamic (evolving in time) factors capturing
  cells' activity.} 
Specifically, in analogy to~(\ref{e.LoadPerType}), we assume that the new
unknowns~$\tilde\theta_j$ are related to each other by  
\begin{align}
\tilde\theta_{j}&=\tilde\theta\frac{P_{j}^{2/\beta}}{P^{2/\beta}},\quad j=1,\ldots,J\, \label{e.TildeThetaj}%
\end{align}
where $P$ is given by~(\ref{e.Power}) and  $\tilde\theta$ solves the following equation
\begin{equation}
\tilde{\theta}=\frac{\rho}{\lambda}\Ex\left[  R^{-1}\left(  \frac
{\frac{1}{L_{X^*}\left(  0\right)  }}{N+\tilde\theta \sum_{j=1}^J\frac{P_{j}^{2/\beta}}{P^{2/\beta}}
\sum_{Y\in\Phi_{j}\setminus\{X^*\} }\frac{1}{L_{Y}\left(
0\right)  }}\right)  \right]\, \label{e.meanFixedPoint}
\end{equation}
%\begin{align*}
%\tilde{\theta}&=\frac{\rho}{\lambda}\Ex\left[  R^{-1}\left(  \frac
%{\frac{1}{L_{X}\left(  0\right)  }}{N+\sum_{j=1}^J\tilde{\theta}_{j}
%\sum_{Y\in\Phi_{j} }\frac{1}{L_{Y}\left(
%0\right)  }}\right)  \right]\\%  \label{e.FixedPoint}\\%
%\tilde\theta_{j}&=\tilde\theta\frac{P_{j}^{2/\beta}}{P^{2/\beta}},\quad j=1,\ldots,J\,. %\label{e.Feller}%
%\end{align*}
The mean fixed point cell load equations boils down hence to an equation in one variable~$\tilde\theta$.
Note that the argument of $R^{-1}$ in~(\ref{e.meanFixedPoint}) is a
functional of the marked path-loss process $\hat\Phi$ and thus 
the expectation in this expression can be evaluated using the equivalent homogeneous
model described in Remark~\ref{r.equivalence}.

By the {\em mean cell of type} $j=1,\ldots,J$ we understand a (virtual) processing
sharing queue with the traffic demand
$$\tilde{\rho}_j:=\bar{\rho}_j=\frac{\rho}{\lambda}\frac{P_j^{2/\beta}}{P^{2/\beta}}\, $$
and the traffic load $\tilde\theta_j$ given by~(\ref{e.TildeThetaj}),
where $\tilde\theta$ is the solution of~(\ref{e.meanFixedPoint}).
The remaining mean cell characteristics (the critical load, user
throughput  and the number of users) are related to these two ``primary''
characteristics in analogy to~(\ref{e.CriticalTraffic2}),
(\ref{e.UserThroughput}) and (\ref{e.UsersNumber}) via 
\begin{align}\label{e.menCriticalTraffic2}
\tilde{\rho_c}_j&:=\frac{\tilde\rho_j}{\tilde\theta_j}\,,\\
\tilde r_j&:=\max(\tilde{\rho_{\mathrm{c}}}_j  -\tilde\rho_j  ,0) \label{e.menUserThroughput}\,,\\%
\tilde N_j&:=\frac{\tilde \rho_j  }{\tilde r_j  }\,,
\label{e.meanUsersNumber}%
\end{align}
$j=1,\ldots,J$.

We will also consider a (global) {\em mean cell} having, respectively,
the traffic demand and cell load given by 
$\tilde{\rho}:=\bar{\rho}=\frac{\rho}{\lambda}$ and $\tilde\theta$,
and the remaining characteristics $\tilde\rho_c,\tilde r,\tilde N$
given, respectively  by~(\ref{e.menCriticalTraffic2}),
(\ref{e.menUserThroughput}) and (\ref{e.meanUsersNumber}) where 
the subscript $j$ is dropped.

In the next section we shall evaluate the mean cell approximation (both globally and per type)
by comparison to the characteristics of the typical cell obtained both from
simulation and from real field measurements.
%In particular we shall check whether the average user's throughput $\bar{r}_{i}$ is well aproximated by \[ \tilde{r}_{i}=\max\left(  E\left[  R^{-1}\left(  \mathrm{SINR}\left( 0,\Phi\right)  \right)  \right]  -\frac{\rho a_{i}}{\lambda_{i}a},0\right) \] where we use~(\ref{e.TrafficPerType}) and~(\ref{e.LoadPerType}).

\section{Numerical results and model validation}
\label{s.NumericalResults}

In this section we present numerical results of the analysis of our
model and compare them to the corresponding statistics obtained from
some real field measurements. We show that the  obtained results match
the real field measurements.
Our numerical assumptions, to be presented in Section~\ref{sss.NumercialAssumptions},
correspond to an operational network in some big city in Europe in
which two types of BS can be distinguished, conventionally called 
macro and micro base stations.~\footnote{Let us explain what we mean
here  by macro and micro BS:
Historically, the operator deployed first what we call here macro BS. 
Powers of these stations slightly vary around some
mean value as a consequence of some local adaptations. We assume
them constant. 
%For example the base station power may be reduced in order
%to reduce the interference caused to a neighboring base station or to reduce
%the amount of traffic it handles. These power adaptations are made by the
%operator during lifetime of the network. 
In order to cope with the increase of the traffic demand, new stations are 
added progressively. These new stations, which we call micro BS, emit with 
the power about 10 times smaller than the macro
BS. Figure~\ref{f.PowerCDF} shows the cumulative distribution
function (CDF) of the antenna powers (without antenna gains).}

\begin{figure}[h]
\begin{center}
\includegraphics[width=0.8\linewidth, height=0.5\linewidth]%
{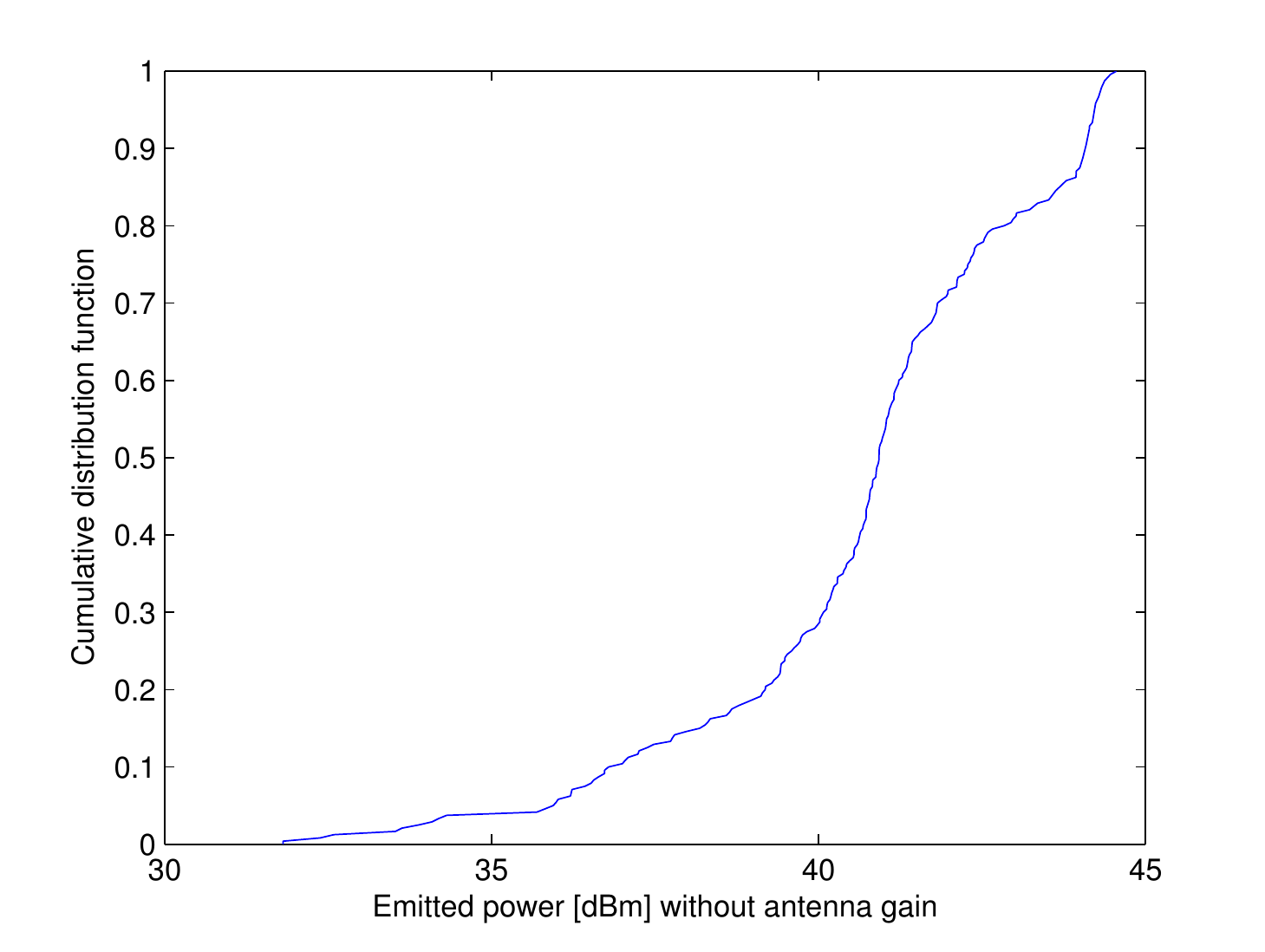}%
\vspace{-2ex}
\caption{CDF  of the emitted antenna powers, without antenna gains, in the
 network consisting of micro (power $\le35$dBm) and macro BS (power
 $\ge 35$dBm). The average micro BS power is $33.42$dBm and the
 average macro BS power is  $41.26$dBm. Adding antenna gains, which  are equal respectively $14$dBm and
 $17$dBm for micro and macro BS, we obtain
 $P_{2}=47.42$dBm,  $P_{1}=58.26$dBm.
\label{f.PowerCDF}}%
\end{center}
\vspace{-3ex}
\end{figure}
%EndExpansion

 The real field measurements are obtained using a methodology described in Section~\ref{sss.Measurements}.
In Section~\ref{ss.Results} the statistics obtained from these
measuremetns will be compared to the performance of each category of
BS calculated using the approach proposed in the present paper.

\subsection{Model specification}

\subsubsection{Measurements} 
\label{sss.Measurements}
%We describe now briefly the performed  real-field measurements. 
The raw data are collected
using a specialized tool used by network maintenance engineers.
This tool has an interface allowing to fetch values of  several
parameters for  every base
station 24 hours a day, at a time scale of one hour. For a given day, for every hour, we obtain information 
regarding the BS coordinates, type, power, traffic demand, number of users
 and cell load calculated as  the percentage of the Transmission Time
Intervals (TTI) scheduled for transmissions. Then we estimate the global cell performance
metrics for the given hour averaging over time (this hour) and next
over all considered network cells. The mean user throughput is calculated as the
ratio of the mean traffic demand to the mean number of users. 
The mean traffic demand $\rho$ is also used as the input of our model. 
Knowing all  cell coordinates, their  types and the surface of the
deployment region  we deduce the network density and fraction of BS in
the two tiers.

\subsubsection{Numerical assumptions}
\label{sss.NumercialAssumptions}
The BS locations are generated as a realization of a
Poisson point process of intensity $\lambda=\lambda_{1}+\lambda_{2}%
=4.62$km$^{-2}$ (which corresponds to an average distance between two base
stations of $0.5$km) over a sufficiently large observation window which is
taken to be the disc of radius $2.63$km. The ratio of the micro to macro BS
intensities equals  $\lambda_{2}/\lambda_{1}=0.039$.
The transmitted powers by macro and micro BS are $P_{1}=58.26$%
dBm,\ $P_{2}=47.42$dBm, respectively. The power~(\ref{e.Power}) of the ``equivalent''
homogeneous network is $P=58.03$dBm.
The propagation loss due to distance is $l(x)=(K\left\vert x\right\vert
)^{\beta}$ where $K=7117$km$^{-1}$ and the path loss exponent $\beta=3.8$.
Shadowing is assumed log-normally distributed with standard deviation
$\sigma=10$dB and spatial correlation $0.05$km.
The technology is HSDPA (High-Speed Downlink Packet Access) with MMSE (Minimum
Mean Square Error) receiver in the downlink. The peak bit-rate equals to
$30\%$ of the information theoretic capacity of the Rayleigh fading channel
with AWGN; that is
\begin{equation}\label{e.Shannon}
R\left(  \mathrm{SINR}\right)  =0.3W\Ex\left[  \log_{2}\left(  1+\left\vert
H\right\vert ^{2}\mathrm{SINR}\right)  \right]
\end{equation}
where the expectation $\Ex\left[  \cdot\right]  $ is with respect to the
Rayleigh fading $H$\ of mean power $\Ex[\left\vert H\right\vert ^{2}]=1$, and
$W=5$MHz\ is the frequency bandwidth.
A fraction $\epsilon=10\%$ of the transmitted power is used by the pilot
channel (which is always transmitted whether the BS serves users or
not).%
~\footnote{It is taken into account by replacing
$\min(\theta(Y),1)$ in~(\ref{e.FixedPoint}) by $\min(\theta(Y),1)(1-\epsilon)
+\epsilon$. Similar modification concerns $\tilde{\theta}$ in the right-hand-side
of~(\ref{e.meanFixedPoint}).}
The
antenna pattern is described in~\cite[Table A.2.1.1-2]{3GPP36814-900}. The
noise power is $-96$dBm.

\subsection{Results}
\label{ss.Results}

We present now the results obtained form the analysis of our two-tier Poisson
 model conformal to a given region of the operational network, adopting both the typical cell approach described in
Section~\ref{ss.TypicalCell} and the mean cell approach explained in Section~\ref{ss.MeanCell}. 
The obtained results are compared to the respective quantities
estimated in the given operational network.
Error bars on all figures represent the standard deviation in the
averaging over 10 realizations of the Poisson network in the
Monte-Carlo estimation of the respective expectations.

Figure~\ref{f.PonderedInterference_Load} shows the mean cell load  together with the stable fraction of
    the network, both globally and separately for the two tiers, as functions of
the mean traffic demand per cell $\bar\rho=\rho/\lambda$. 
The mean cell load is calculated using the two  approaches: the
typical cell and the mean cell one. The stable fraction of the network
is available only in the typical cell approach. 
Figure~\ref{f.PonderedInterference_Load} presents also 24 points
presenting the mean cell load estimated from the real field measurements done during 24 different hours of some given day.
Note a good fit of our results and the network measurements.
Observe also that all real field measurements fall within the range
of the traffic demand ($\bar\rho\le 600$kbps) for which the stable fraction of the network for
both network tiers is very close to~1. This is of course a consequence of
a good dimensioning of the network. Interestingly these latter metrics
allows us to reveal  existing dimensioning margins. Specifically, we
predict that there will be no unstable macro cells with the traffic
demand slightly less than $\bar\rho\le 700$kbps and the micro cells remain stable for much
higher traffic demand of order  $\bar\rho\approx1000$kbps.

We move now to the mean number of users per cell presented on
Figure~\ref{f.PonderedInterference_UsersNumber}, again,  as function of
the mean traffic demand per cell $\bar\rho=\rho/\lambda$. 
Both approaches (typical and mean cell) are adopted and the two
network tiers are analyzed jointly and separately.
As for the load, we  present also 24 points corresponding the network
measurements. Note a good fit of our model results and the network
measurements.
Note also that the prediction of the model performance  for
the traffic demand $\bar\rho\ge 700$kbps , where the fraction of
unstable cell is non-negligible (cf
Figure~\ref{f.PonderedInterference_Load}), is much more
volatile. More precisely, the relatively large error-bars
of the mean number of users for $\bar\rho\ge 700$kbps can be explained
by a non-negligible probability of finding a cell whose load is
just below~1. It still contributes to the calculation of the mean
number of users (as we remove only strictly unstable cells) and makes
the empirical mean very large for this simulation experiment.

Finally, Figure~\ref{f.PonderedInterference_UserThroughput} shows the
relation between  the  mean user throughput and 
the mean traffic demand per cell $\bar\rho$ obtained via the two modeling
approaches and real field measurements. The global performance of the
network and its macro-tier  are quite well captured by our two
modeling approaches. The micro-tier analysis via the typical cell and the real
field measurements exhibit important volatility due to 
a relatively small number of such cells in the network.
The mean cell model allows to predict however a macroscopic law in
this regard.

\begin{figure}[t!]
\begin{center}
\hspace{-1em}%
\includegraphics[width=0.9\linewidth
]%
{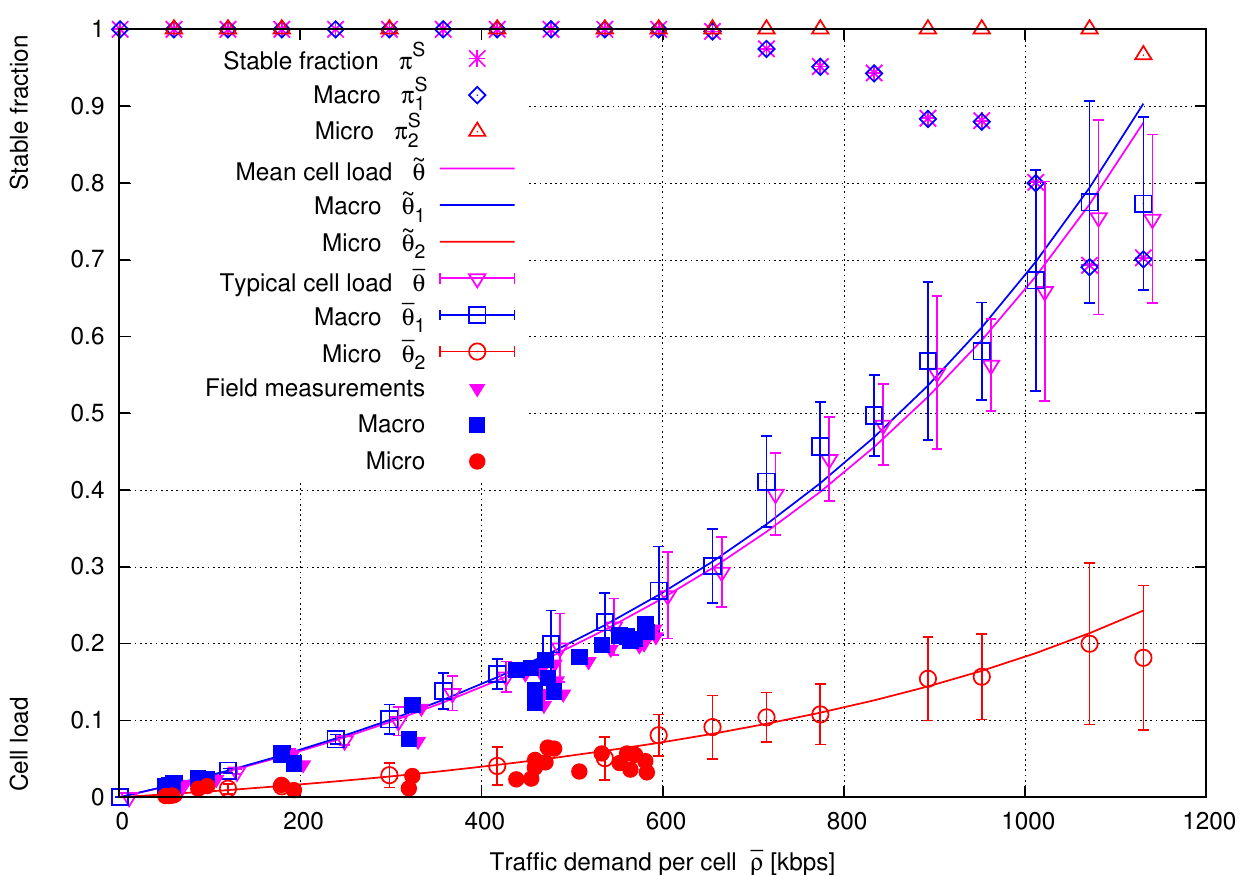}%
\vspace{-2ex}
\caption{Cell load versus traffic demand per cell.}%
\label{f.PonderedInterference_Load}%
\end{center}
%\end{figure}
%EndExpansion
%

%TCIMACRO{\FRAME{fpFU}{3.5276in}{2.4735in}{0pt}{\Qcb{Number of users per cell
%versus traffic demand per cell in the Pondered interference model.}%
%}{\Qlb{f.PonderedInterference_UsersNumber}}%
%{f.PonderedInterference_UsersNumber}{\special{ language "Scientific Word";
%type "GRAPHIC";  maintain-aspect-ratio TRUE;  display "ICON";
%valid_file "F";  width 3.5276in;  height 2.4735in;  depth 0pt;
%original-width 5.0004in;  original-height 3.4938in;  cropleft "0";
%filename 'SimulationsPIM/PonderedInterference_UsersNumber';file-properties "XNPEU";}%
%}}%
%BeginExpansion
%\begin{figure}[p]
\begin{center}
\hspace{-1em}%
\includegraphics[width=0.9\linewidth
]%
{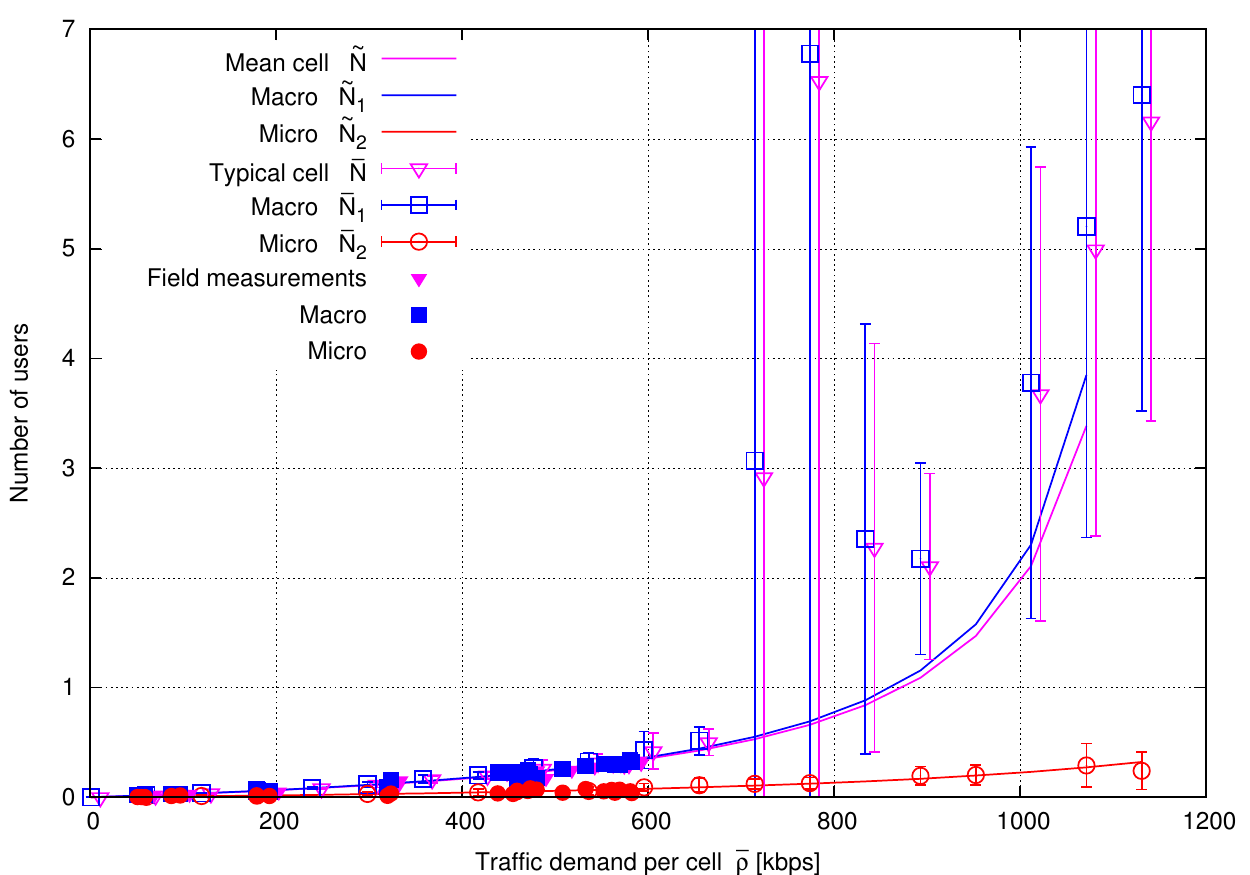}%
\vspace{-2ex}
\caption{Number of users per cell versus traffic demand per cell.}%
\label{f.PonderedInterference_UsersNumber}%
\end{center}
%\end{figure}
%EndExpansion
%

%TCIMACRO{\FRAME{fpFU}{3.5276in}{2.4735in}{0pt}{\Qcb{Mean user throughput in
%the network versus traffic demand per cell in the Pondered interference
%model.}}{\Qlb{f.PonderedInterference_UserThroughput}}%
%{f.PonderedInterference_UserThroughput}{\special{ language "Scientific Word";
%type "GRAPHIC";  maintain-aspect-ratio TRUE;  display "ICON";
%valid_file "F";  width 3.5276in;  height 2.4735in;  depth 0pt;
%original-width 5.0004in;  original-height 3.4938in;  cropleft "0";
%filename 'SimulationsPIM/PonderedInterference_UserThroughput';file-properties "XNPEU";}%
%}}%
%BeginExpansion
%\begin{figure}[p]
\begin{center}
\hspace{-1em}%
\includegraphics[width=0.9\linewidth
]%
{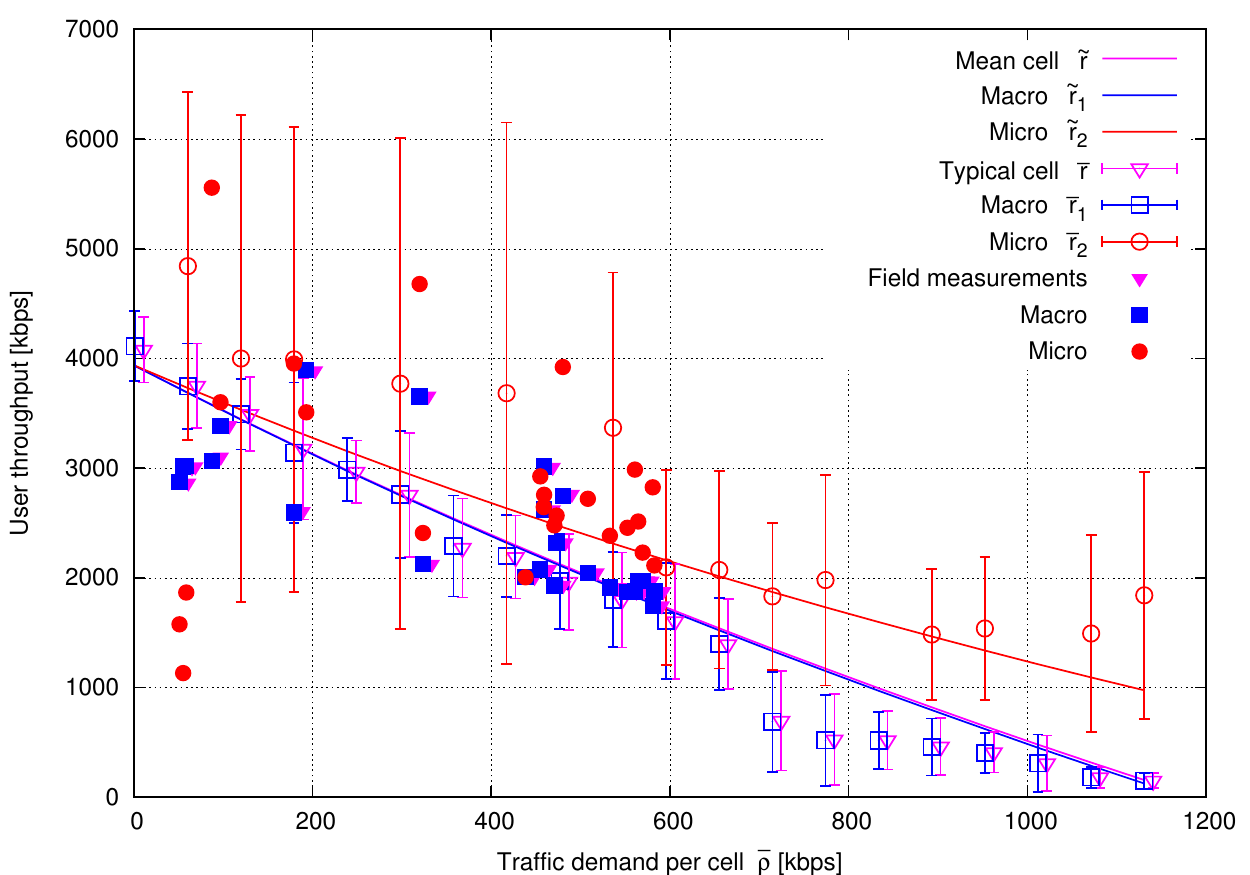}%
\vspace{-2ex}
\caption{Mean user throughput in the network versus traffic demand per cell.}%
\label{f.PonderedInterference_UserThroughput}%
\end{center}
\vspace{-5ex}
\end{figure}
%EndExpansion

\section{Conclusions}
A heterogeneous  cellular network model  allowing for different
BS types (having different transmission powers) is proposed, aiming to help in 
performance evaluation and dimensioning of real (large, irregular) operation networks.
It allows one to identify  key laws relating the performance of the different base station
types. In particular, we show how  the mean load of different types of
BS is related in a simple way to their  transmission powers. 
%This law shows
%that, contrarily to intuition, the base stations transmitting low powers have
%better performance than those with high powers.
The results of the model analysis are compared to
real field measurement in an operational network showing its pertinence.

\addtocounter{section}{1}
\addcontentsline{toc}{section}{References} 
{\small
\bibliographystyle{IEEEtran}
%\bibliography{HetNets}

}
\end{document}